\RequirePackage{amsmath,amsfonts}
\documentclass{article}

\usepackage{tikz}

\usepackage{graphicx}
\usepackage{csquotes}
\usepackage{xcolor}

\usepackage{amsmath,amsthm}

\newtheorem{theorem}{Theorem}
\newtheorem{observation}{Observation}

\theoremstyle{definition}
\newtheorem{remark}{Remark}
\newtheorem{example}{Example}
\newtheorem{definition}{Definition}
\newtheorem{algorithm}{Algorithm}

\DeclareMathOperator*{\argmin}{arg\,min}
\DeclareMathOperator*{\argmax}{arg\,max}
\newcommand{\pref}{\succ}

\newcommand{\myheading}[1]{\medskip\noindent\textbf{#1}}

\newcommand{\iti}{{\it i}}
\newcommand{\itii}{{\it ii}}

\begin{document}

\title{Reality-Aware Social Choice\thanks{%
  A preliminary version of this paper was presented at
  the BlueSky track of
  the 17th International Conference on Autonomous Agents and Multiagent Systems (AAMAS '17)~\cite{preliminary}.}
}

\author{Ehud Shapiro \and Nimrod Talmon}


\maketitle

\begin{abstract}
Social Choice theory generalizes voting on one proposal to ranking multiple
proposals.  Yet, while a vote on a single proposal has the status quo (Reality) as a default, Reality has been forsaken during this generalization. Here, we propose to restore this default social state and to incorporate Reality explicitly into Social Choice.
We show that doing so gives rise to a new theory, complete with its domain restrictions, voting rules with their Reality-aware axiomatic properties, and certain game-theoretic aspects. In particular, we show how Reality can be used in a principled way to break Condorcet cycles and develop an efficient Reality-aware Condorcet-consistent agenda. We then discuss several applications of Reality-Aware Social Choice.
\end{abstract}


\section{Introduction}\label{section:introduction}

When voting on a proposal one in fact chooses between two alternatives: (\iti)
  A new hypothetical social state depicted by the proposal and (\itii) the status quo or the present reality (henceforth: \emph{Reality});
a \emph{Yes} vote favors a transition to the proposed hypothetical state, while a \emph{No} vote favors Reality.\footnote{
We capitalize \emph{Reality} to denote that it is a variable that has a value  --- the present state of affairs --- that can change over time, as as reality changes due to our decisions and actions, as well as external events.  See discussions below on Possible Worlds Semantics and Democratic Action Plans. Furthermore, we prefer Reality over the status quo to emphasize our view that this social state is evolving.}

The standard model of Social Choice, while generalizing from voting on one proposal to voting on several proposals, does not give any special consideration to the present social state.
  For example, neither Arrow~\cite{arrowbookfirst}, nor Sen~\cite{sen1986social}, nor Black~\cite{duncan1958theory} incorporate Reality into their models;
the same holds for the textbook on Computational Social Choice~\cite{moulin2016handbook}.
Here we propose a new model of Social Choice that
incorporates Reality explicitly as an \emph{ever-present}, \emph{always-relevant}, and \emph{evolving} social state, \emph{distinguished} from hypothetical social states.

\newpage

\myheading{Distinguished --}
In the words of Arrow's later-years comment: 
\begin{displayquote}
\emph{``... an important empirical truth, especially about legislative matter rather than the choice of candidates:
The status quo does have a built-in edge over all alternative proposals.''}~\cite[Page 95]{arrowbooksecond}
\end{displayquote}

James Buchanan, another Social Choice luminary, stated that: 
\begin{displayquote}
``\emph{The uniqueness of the status quo lies in the simple fact of its existence. [...]  This elementary distinction between the status quo and its idealized alternatives is often overlooked. [...] Any proposal for change involves the status quo as the necessary starting point.}''~(\cite{buchanan1975limits}, p.75)
\end{displayquote}

\myheading{Ever-present --}
since every voter should have the opportunity to prefer the status quo over some or all other proposed alternatives. This opportunity is necessarily available when voting on a single alternative---chosen by voting \emph{No}---and there can be no justification for revoking this right when the number of alternatives increases. Yet, the standard model of Social Choice, by failing to incorporate the status quo as an ever-present alternative, effectively allows those who set up a vote to veto the status quo. 

For example, the governing body may decide unilaterally that a change is needed and offer a closed list of proposed changes to the public.  Then,  seemingly-democratically, proceed with the winning proposal, but without allowing the public to oppose the governing body's unilateral decision that a change is needed or, equivalently, without allowing the voters to oppose all government-sponsored proposals for change, by voting for the status quo.

In this sense, the requirement to incorporate the status quo as an alternative when voting on a closed list of proposals is reminiscent of Dodgson's suggestion to incorporate a fictitious protest candidate in elections, in order to allow a voter to oppose all proposed candidates~\cite{dodgson1873discussion}.  In the same vein, Clark Kerr has aptly stated more than half a century ago that:
\begin{displayquote}
``\emph{the status quo is the only solution that cannot be vetoed.}''~\cite{kerr2001uses} 
\end{displayquote}
Indeed, failing to incorporate the status quo as an ever-present alternative effectively allows those who set up a vote the choice to veto the status quo.

With hindsight, we find it peculiar that such a basic desideratum, arguably a basic civic right, could have been bypassed for so long by Social Choice theory. Indeed, Social Choice theory considers also situations which are unrelated to civic rights,
such as, e.g., helping a group of people to select a restaurant for dinner, or select which pizza to order. Our view is that even in such situations, a default or opt-out alternative must be provided, e.g. stay home, when selecting from a closed list of restaurants, or stay hungry, when selecting from a closed list of dishes.  Not offering such a default must be an indication that those who set up the vote have decided, possibly non-democratically, that opting out is not an option.

\myheading{Always-relevant --}
as the ranking of hypothetical social states critically depends on Reality:
Real voters do not rank hypothetical social states in the abstract; they live in a particular Reality, thus when ranking alternatives to it they in fact rank transitions from Reality to the hypothetical social states depicted by the alternatives. 

As stated by Buchanan: 
\begin{displayquote}
``\emph{... the choice is never carte blanche. The choice among alternative structures [...] is between what is and what might be.}''~(\cite{buchanan1975limits} \emph{ibid}) 
\end{displayquote}

Rational evaluation of such a transition from Reality requires comparing the utility of the hypothetical social state to that of the present Reality, and estimating the cost of realizing the hypothetical social state given Reality.  

\myheading{Evolving --}
as the present social state---Reality---changes over time, and could change even when the set of alternatives remains intact. E.g., democratic communities endeavor that such changes be commensurate with the will of their people. This view can be the basis for \emph{democratic action plans} (Section~\ref{section:action plans}),
which start with the present Reality and proceed only in transitions that have sufficient support.


\begin{example}
As a concrete example of the role of Reality in social choice, consider choosing a location for a public building such as a hospital.  We first consider the scenario of moving an existing hospital to a new location.  Clearly, if only one location is proposed, then the choice is between moving the hospital and retaining it in its existing location.  Offering more than one location should not be used as a pretense to ``veto the status quo'' (cf. Kerr), hence, according to Reality-Aware Social Choice, the option to retain the status quo should be present even if multiple new locations are offered for a vote.
Furthermore, voter preferences might change depending on the present location of the hospital.

Now consider the scenario of electing a location for a new hospital. If only one location is offered for a vote, then clearly the vote would be between this location and no location.  Hence, similarly to the scenario of relocating a hospital, ``no location'', or ``protest'' (see discussion of Protest candidate below) should still be one of the alternatives offered when more than one location is put for a vote.  
\end{example}

\begin{example}
One scenario to consider is novel legislation in an area where there is none, with the extreme example being the writing of a constitution when a political system abandons a
dictatorial regime and attempts to become a democratic one.  
If the body that proposes this new legislation aims to obtain democratic support for it, then the alternative against which the proposal is voted upon is the present reality, in which there is no legislation or no constitution.  Even if several alternatives are presented to voters to choose from, our arguments above dictate that the null alternative, of no legislation or no constitution, must also be present.  Should the null alternative be chosen, the proposers can go back to the drafting board and come back with better proposals that may win democratic support.
\end{example}

\subsection{Related Work}
Richelson~\cite{richelson1984social} proposes axioms requiring that a winning state shall win over the status quo by a majority;
discusses corresponding voting rules;
and considers repeated elections where the status quo might change.
He also provides the following example demonstrating how ignoring the status quo might result in intuitively unaccepted choices.

\begin{example}[\cite{richelson1984social}]
Consider an election over $\{a, b, c, s\}$ with $s$ being the status quo,
and the following votes:
  $3$ votes with $a \pref s \pref b \pref c$,
  $4$ votes with $c \pref s \pref a \pref b$,
  and
  $2$ votes with $b \pref s \pref c \pref a$.
Then,
Plurality would choose $c$ as the winner, even though a majority prefers the status quo ($s$) over $c$.
\end{example}

Grofman~\cite{grofman1985neglected} considers the status quo in the context of spatial voting,
where people correspond to ideal positions on a line
and
vote by averaging their ideal position on that line with the position of the status quo;
thus, in particular, when Reality changes, their preferences change as well.

Gibbard et al.~\cite{gibbard1987arrow} and Yanovskaya~\cite{yanovskaya1991social} also consider \emph{distinguished} social states; specifically, they study social choice functions in which some predefined set of alternatives must be contained in any feasible set.
In their model Reality can be given as a distinguished alternative, however it does not change, but is merely present as an indistinguishable alternative.

We also mention work on the \emph{status quo effect}~\cite{bowen2014mandatory}.

Lastly,
we mention the Amendment agenda~\cite{moulin2016handbook},
employed by Anglo-American legislatures,
in which voting is carried out sequentially on each amendment given the status quo,
which is continuously updated by each approved amendment.

\subsection{Structure of the Paper}

Section~\ref{section:formal} formally describes the model of Reality-aware Social Choice. In doing so we discuss Reality-aware Social Choice Functions (Section~\ref{section:rscf}) as well as voter preferences (Section~\ref{section:preferences}) including certain domain restriction on such Reality-aware voter preferences (Section~\ref{section:domain restrictions}).

Section~\ref{section:worlds} discusses the underlying semantics of possible worlds  together with their related game-theoretic aspects of democratic action plans (Section~\ref{section:action plans}). Section~\ref{section:condorcet}  discusses Reality-aware Condorcet criteria including a description of an efficient Reality-aware Condorcet-consistent agenda (Section~\ref{section:agenda}).
Section~\ref{section:discussion}, end the paper with a discussion of applications of Reality-Aware Social Choice (Section~\ref{section:applications}) and further research directions stemming from this work.

We believe the main contribution of the paper is the introduction of a formal model of Reality-aware social choice, which consists of an adaptation of social choice functions and a novel definition of Reality-aware voter preferences. The formal model is accompanied with a semantics of possible worlds and a formal game which is relevant for a particular type of domain restriction of voter preferences. As we demonstrate and discuss in Section~\ref{section:applications}, several applications emerge from the model of Reality-Aware Social Choice described here.

\section{Formal Model of Reality-Aware Social Choice}\label{section:formal}

We first discuss our main definition,
that of Reality-Aware Social Choice Functions (RASCFs).
Later in this section we discuss Reality-aware voter preferences together with certain domain restrictions on these types of preferences.

\subsection{Reality-Aware Social Choice Functions (RASCFs)}\label{section:rscf}

In the standard model of social choice, the pivotal definition of a social choice function (SCF) is as follows.

\begin{definition}[SCF]
Let $A$ be a set of alternatives and let $L(A)$ be the set containing all linear orders on $A$.
Then, a social choice function (SCF) is a function $f : L(A)^n \to A$.
\end{definition}

That is, an SCF takes $n$ votes, each being a linear order over the same set of alternatives, and returns one alternative as the winner of the election.
Here, we extend this definition and define Reality-aware social choice function (RASCF), to take not only $n$ votes but also a Reality as an input, as defined next.

\begin{definition}[RASCF]
Let $A$ be a set of alternatives and let $L(A)$ be the set containing all linear orders on $A$.
Then, a \emph{Reality-aware social choice function} (RASCF) is a function $f : L(A)^n \times A \to A$.
\end{definition}

\begin{remark}
It is possible to similarly define a Reality-aware Social Welfare Functions.
For focus considerations, we however do not consider such functions in this paper.
\end{remark}

\subsection{Reality-aware Voter Preferences}\label{section:preferences}

The definition given above of Reality-aware social choice functions is our basic definition.
Notice, however, that it considers a specific Reality one at a time.
Below we are concerned with voter preferences as they relate to a changing Reality.
That is, we assume that voters have possibly different preferences for different realities, and we are concerned with how their preferences change as the reality changes.
Our basic definition of voter preferences in our model is as follows.

\begin{definition}[Reality-aware voter preferences]\label{definition:voterpreferences}
Given a set of alternatives $A$, the Reality-aware voter preferences of voter $v$ are defined via a function $f_v$ such that $f_v : A \to L(A)$.
\end{definition}

That is, for each possible Reality $R$, $f_v(R)$ is the linear order representing the preferences of voter $v$ if and when the Reality is $R$. Notice the difference with the standard model of Social Choice,
in which voter preferences are given as a single linear order;
in our model, a voter is associated not only with one linear order,
but with several, one for each possible Reality.

\begin{remark}
Indeed,
here we consider the ordinal model of elections,
where voters provide linear orders.
Adapting other elicitation methods into our model is of course possible but we do not discuss such adaptations here.
\end{remark}

\begin{remark}
To the best of our knowledge, there is no natural way to reduce our model of Reality-aware Social Choice to the standard setting of social choice. This is especially apparent in view of Definition~\ref{definition:voterpreferences} of Reality-aware voter preferences. Indeed, there does not seem to be a natural reduction of this type of preferences to standard ordinal preferences.
\end{remark}

\subsection{Domain Restrictions}\label{section:domain restrictions}

While it might make sense to consider arbitrary Reality-aware voter preferences,
next we consider several domain restrictions,
by restricting those functions which define the voter linear order for each possible Reality.
Indeed, formally we view the following as domain restrictions,
as they restrict the allowed preferences of the voters.
We also mention that some of these restrictions give rise to certain models by which voter preferences can be constructed
(similarly to, say, single-peaked preferences in the standard model of Social Choice).
Other domain restrictions can be defined, of course. Here we aim to delve into the structure of Reality-aware voter preferences to gain further understanding of them and the demonstrate the richness of the Reality-aware Social Choice model. A certain usage of one of the domain restrictions we consider below is discussed in Section~\ref{section:action plans}.

\subsubsection{Abstract Constancy Preferences}

We begin with the least domain restriction, which requires the following kind of constancy: If one craves for a particular state of affairs independently of Reality, then once this wish is fulfilled and this state becomes Reality, the craving should be satisfied.

For a Reality $R$,
an alternative $s$,
and a voter $v$,
we denote the position of alternative $s$ in the linear order $f_v(R)$ by $pos_{f_v(R)}(s)$.
For example, if $v$ prefers $s$ the most when the Reality is $R$,
then $pos_{f_v(R)}(s) = 1$.

\begin{definition}[Abstract Constancy Preferences]\label{constancy}
We say that the preferences of voter $v \in V$ satisfy \emph{abstract constancy} if the following hold for each social state $s\in S$: If $pos_{f_v(R)}(s) = 1$ for each $R \in S \setminus \{s\}$,
  then $pos_{f_v(s)}(s) = 1$.
  
  A Reality-aware election satisfies abstract constancy if the preferences of each voter $v \in V$ satisfy abstract constancy.
\end{definition}
In other words, abstract constancy requires that if a state $s$ is ranked first in $f_v(R)$ for any $R \in S \setminus \{s\}$, then $s$ is also ranked first in $f_v{s}$.

\begin{remark}
Notice that, if there is only one possible Reality,
then the above restriction holds (in the empty sense, since $S \setminus \{s\} = \emptyset$).
In this sense, voter preferences in the standard model of Social Choice satisfy abstract constancy.
Furthermore,
even if there are several possible Realities,
then even though $S \setminus \{s\} = \emptyset$,
simply letting $f_v(s) = f_v(s')$ for each voter $v \in V$ and each pair of social states $s$, $s'$
gives a Reality-aware election that satisfies abstract constancy.
Again,
it follows that voter preferences in the standard model of Social Choice satisfy abstract constancy.
\end{remark}

Observe that the definition above is indeed a domain restriction,
as (similarly to, say, single-peakedness in standard Social Choice)
it restricts the (Reality-aware) voter preferences.
We view this restriction as a sanity check,
namely as a very weak restriction which usually shall be satisfied.

Next we consider further, more restricted domain restrictions.

\subsubsection{Distance Constancy Preferences}

Below we consider a further restriction on the domain of Reality-aware voter preferences, based on the social cost of transitioning from Reality to hypothetical social states.
In addition to voters~$V$, social states $S$ with changing Reality $R$, and rankings $v_R$, we assume that there is a \emph{distance} $d$ over $S$. Formally, it is a pseudoquasimetric $d : S \times S \to \mathbb{N}$,
satisfying
(1) $d(s, s') \geq 0$ for all $s, s' \in S$;
(2) $d(s, s) = 0$ for all $s \in S$;
and
(3) $d(s, s'') \leq d(s, s') + d(s', s'')$ for all $s, s', s'' \in S$.  
For two social states $s$, $s'$ the value $d(s, s')$ can be viewed as the distance between $s$ and $s'$ in some abstract metric space, or as the \emph{social cost} of transitioning from $s$ to $s'$.
Below we consider Reality-aware voter preferences which are generated from such a model; and we require that a voter preferring a particular state from afar should prefer it at least as much as it gets near.

\begin{definition}[Distance Constancy Preferences]
Given a pseudoquasimetric $d : S \times S \to \mathbb{R}$, we say that the Reality-aware preferences of a voter $v \in V$ satisfy \emph{distance constancy} with respect to $d$ if the following hold for each triplet of social states $R, a, b \in S$: If $pos_{f_v(R)}(a) < pos_{f_v(R)}(b)$,
  then for each $R' \in S$ for which $d(R', a) \leq d(R, a)$ and $d(R', b) \geq d(R, b)$,
  it holds that $pos_{f_v(R')}(a) < pos_{f_v(R')}(b)$.
  
We say that a Reality-aware election satisfies \emph{distance constancy} if there is a pseudoquasimetric $d$ for which every voter $v \in V$ satisfies distance constancy.
\end{definition}
In other words, distance constancy requires that if $a$ is ranked above $b$ in $f_v(R)$,  then $a$ is ranked above $b$ in $f_v(R')$ for every $R' \in S$ for which $d(R', a) \leq d(R, a)$ and $d(R', b) \geq d(R, b)$.

Next we show that this domain restriction is at least as restrictive as the previous domain restriction.

\begin{observation}\label{observation:constancylemmaone}
 If voter preferences satisfy  distance constancy then they satisfy  abstract constancy.
\end{observation}

\begin{proof}
By way of contradiction, assume there is a voter $v \ in v$ with preferences that satisfy distance constancy for some $d$ but not satisfy abstract constancy.
Then there is a state $s$ for which $s$ is ranked first in $f_v(R)$ for each $R \neq s$ but is not ranked first in $f_v(S)$.
Consider some $s' \neq s$ and notice that $d(s', s) \geq d(s, s)$, thus the fact that $s$ is ranked higher than all other states in $v_s'$ but there is some state which is ranked higher than $s$ in $v_s$ violates distance constancy. A contradiction.
\end{proof}

\subsubsection{Utility Constancy Preferences}

We consider a further restriction on the domain . We assume not only a distance $d$ between the social states as above, but also assume that each voter has a \emph{state utility}, denoted by $U_v(S)$, and require that the distance $d$ of these states from Reality, combined with their utilities, completely define the Reality-aware preferences of the voters.

To gain some intuition for this domain restriction, as well as a specific application, consider selecting a budget~\cite{budgetingone}.
In this case it is reasonable to assume an objective distance between budget bundles
(e.g., the total amount of funds in the symmetric difference between two budget bundles might define such an objective distance).
Then, combined with a utility of a voter to each of the possible budget bundles, the preferences of each voter is determined.
So the distance $d$ to be thought of an objective distance between the states,
while the state utility values are individual for each voter,
and reflect each voter's utility for each state~$s$, which is oblivious to Reality.
(Such an objective distance between social states might not be reasonable for, e.g., electing a winner among a set of candidates.)

Formally, 
we first define a \emph{transition utility},
which we then use to define the utility constancy domain restriction.

\begin{definition}
Given a pseudoquasimetric $d : S \times S \to \mathbb{R}$
and a state utility function $U_v : S \to \mathbb{R}$, the \emph{transition utility} $T_v$ of voter $v$ from state $s$ to $s'$ is defined as follows:
  $T_v(s, s') = U_v(s') - U_v(s) - d(s, s')$. 
\end{definition}

That is, the utility of a certain voter for transitioning from state $s$ to state $s'$
is defined as the difference in the utility of the voter from state $s'$ minus its utility for $s$,
minus the distance to transition from state $s$ to state $s'$.
Intuitively,
a voter would like to move from $s$ to $s'$ if the gain in utility is significant enough
to compensate for the cost of transitioning from $s$ to $s'$.

Then, the preferences of voter $v$ satisfy utility constancy if for each Reality $R$ and two states $s$, $s'$,
it holds that voter $v$ prefers $s$ over $s'$ if the transition utility from $R$ to $s$ is higher than the transition utility from $R$ to $s'$.  That is, the ranking of a voter $v$ with respect to some Reality $R \in S$ is determined by ordering the social states $s \in S$ in decreasing order of their transition utilities from Reality, namely $T_v(R,s)$.
The intuition is that transitioning to a farther social state is more costly and hence less desirable.
A formal definition follows.

\begin{definition}[Utility Constancy Preferences]
Given a distance $d$ and for each voter $v \in V$ a state utility $U_v$, the preferences of a voter $v \in V$ satisfies \emph{utility constancy} with respect to $d$ and $U_v$, if for each Reality $R$ and two states $s$, $s'$, it holds that $pos_{f_v(R)}(s) < pos_{f_v(R)}(s')$ if and only if $T_v(R, s) > T_v(R, s')$.

We say that a Reality-aware election satisfies \emph{utility constancy} if there is a distance $d$ and for each voter $v \in V$ a utility function $U_v$ such that the Reality-aware preferences of every $v \in V$ satisfy utility constancy with respect to $d$ and~$U_v$.
\end{definition}

Next we show that this domain restriction is as restricted as the previous ones.

\begin{observation}
 If voter preferences satisfy utility constancy then they satisfy distance constancy and abstract constancy.
\end{observation}

\begin{proof}
By way of contradiction, assume there is a voter $v \in v$ with preferences that satisfy utility constancy for some distance $d$ and utility $U_v$, but not distance constancy for $d$.
Then there are Realities $R$, $R'$ and states $a$, $b$ for which $d(R', a) \leq d(R, a)$, $d(R', b) \geq d(R, b)$, $pos_{f_v(R)}(a) < pos_{f_v(R)}(b)$,   but $pos_{f_v(R')}(a) > pos_{f_v(R')}(b)$. 

Since $v$ prefers $a$ over $b$ for Reality $R$, it follows that $T_v(R, a) < T_v(R, b)$. Thus, $T_v(R, a) = U_v(a) - U_v(R) - d(R, a) < U_v(b) - U_v(R) - d(R, b) = T_v(R, b)$. It follows that $U_v(a) - d(R, a) < U_v(b) - d(R, b)$. 
Similarly, since $v$ prefers $b$ over $a$ for Reality $R'$, it follows that $T_v(R', a) > T_v(R', b)$ and thus $U_v(a) - d(R', a) > U_v(b) - d(R', b)$. Plugging in the facts that $d(R', a) \leq d(R, a)$ and $d(R', b) \geq d(R, b)$ we reach a contradiction. 
\end{proof}

\section{Possible Worlds and Democratic Action Plans}\label{section:worlds}

In this section we discuss semantics for Reality-Aware Social Choice. Specifically,
in Reality-Aware Social Choice, each social state $s \in S$ can be viewed as a \emph{possible world}
(see, e.g.,~\cite{chellas1980modal}), where Reality $R \in S$ is the state corresponding to a possible world that is the actual world, namely the world in which the vote on $S$ takes place. If a social state $s \neq R$ wins, this means that society aims to change Reality from $R$ to $s$.  If successful, then the subsequent vote may take place in the possible world (i.e., with Reality being) $s$.  
Indeed, in this view, our Reality-Aware Social Choice Functions are functions which take voter preferences with respect to the possible world which we are at now, and elects the possible world to transition to.

This view of possible worlds and the possible transitions among them gives rise to \emph{democratic action plans}, which are (possibly infinite) sequences of states where
(i) the first state is the present Reality;
and
(ii) when Reality equals a state in the sequence, then the following state wins the vote (or ties) in that state.
A democratic action plan need not be executed blindly to completion,
but should be evaluated at each step, accommodating gaps between estimated and actual efforts of transitioning from one state to another, changes of heart, and recent real-world, external events.
Still, a democratic action plan may provide a long term vision and blueprint for a democratic community that is useful and reassuring,
while being consistent with the immediate action to be taken at any point in time.

A democratic action plan offers an incremental path from the present Reality with the property that each incremental choice along the way has democratic support.
Some state transitions may not be feasible.
If a transition from $s$ to $s'$ is infeasible or, using possible-worlds terminology, $s'$ is not accessible from $s$, then $s'$ should not be offered as an alternative when voting in the possible world~$s$.

Computing a democratic action plan requires eliciting from each voter the ranking of all states with Reality being each state in the plan, except the last.  As this may be difficult, one might consider unfolding democratic action plans, in which the society computes only a few steps ahead.
This seems reasonable as a society can anyhow peer so much into the future.  In particular, it may not be plausible to assume that voters can predict what their preferences will be when the new Reality is far away from the present one.

A democratic action plan requires a voting rule that determines when one state is preferred over another,
given the Reality.
We discuss Reality-aware Condorcet-consistent voting rules below, as well as an efficient Reality-aware Condorcet-consistent agenda.

Notice that a democratic action plan,
if unbounded,
might not converge,
as the following example demonstrates.

\begin{example}\label{example:infinite democratic action plan}
Consider three social states: $a$, $b$, and $c$,
and a single voter $v$.
Assume that the ranking of $v$ when Reality is $a$ is $v_a = b \pref c \pref a$; her ranking when Reality is $b$ is $v_b = c \pref a \pref b$; and her ranking when Reality is $c$ is $v_c = a \pref b \pref c$.
Then,
even though our Constancy axiom (i.e., Definition~\ref{constancy}) (vacuously) holds,
we nevertheless have a cyclic,
thus infinite democratic action plan:
  To see this,
  assume that $a$ is the initial Reality.
  Then, the society would move to $b$,
  then to $c$,
  and back to $a$.
\end{example}

\subsection{Game-Theoretic Consequences}\label{section:action plans}

Here we discuss game-theoretic consequences of our model. Specifically, Reality-aware Social Choice, supplied with such democratic action plans as discussed above,
allows for a more realistic study of strategic behavior of voters.
For concreteness, consider a utility constancy election in which the objective distance~$d$ and the individual state utilities $U_v$ are known.
This suggests an iterative game, pictorially played on the following graph
(this game might be thought of a strategic game on top of a democratic action plan over
the model of Utility-based Reality-aware Social Choice):
  The \emph{game graph} $G$ has a vertex $a$ for each social state $a \in S$ and is a complete arc-weighted directed graph,
where the weight of the arc $(a, b)$ equals the distance from $a$ to $b$, according to $d$
(i.e., $w(a, b) = d(a, b)$ for each $a, b \in S$).
In each turn,
the current Reality is represented by a pebble placed on one vertex,
and each voter is a player.
Assume,
for simplicity,
that the players know the metric $d$ and the state utilities of all players.
The game is played repeatedly,
where in each turn all players specify their strategic rankings which might not correspond to their real transition utilities,
as the pay-off of each player equals the state utility of that player from the Reality at the end of the game\footnote{%
  In fact, a more careful formulation of the game might be that it is played until convergence;
  or, that the players' pay-off is averaged over the course of the game.}
(so behaving strategically may, e.g., help shift society to a social state this player prefers less,
in the hope of more easily being able to shift the society from that state to a more preferable state).


\section{Reality-Aware Condorcet-Consistency}\label{section:condorcet}

Next we incorporate Reality into two classical concepts -- the $18^{th}$-century Criterion of the Marquis de Condorcet and the $12^{th}$-century Amendment Agenda of Ramon Lull.  First, we discuss how incorporating Reality into Condorcet criteria can cope with Condorcet cycles in a natural way, resulting in simple one-sentence voting rules that can be easily explained to voters.  Then, we show how incorporating Reality into the Amendment Agenda can produce a simple and efficient Reality-aware Condorcet consistent voting process, that can be realized simply by show of hands and computer-less note taking.

\subsection{Reality-Aware Axiomatic Properties}\label{section:reality shock condorcet}

Here we consider several axiomatic properties of voting rules that use Reality.
As our voting rules are defined via Reality-aware Social Choice Functions,
we use the term Reality-aware voting rule to formally mean a Reality-aware Social Choice Function. In this section we focus on axiomatic properties related to the Condorcet principle.

The first use of Reality should be to eliminate all alternatives not preferred over it -- why bother with an alternative that the majority considers inferior to the status quo? 

\begin{definition}[Reality-viable Alternatives]\label{definition:reality-viable}
Let $S$ be a set of alternatives and $R \in S$ the Reality.
An alternative $s \in S$ is \emph{Reality-viable} if $s$ beats $R$ (i.e., if more voters prefer $s$ over $R$ than vice versa).  The set of Reality-viable alternatives is denoted by $S_R$.
\end{definition}

Correspondingly, what is perhaps the simplest axiomatic property desired from a voting rule that uses Reality is that it shall elect a Reality-viable alternative, if there is one, else elect Reality.

\begin{definition}[Reality-viable Criterion]\label{definition:reality viable}
Let $S$ be a set of alternatives with $R \in S$ the Reality. A Reality-aware voting rule satisfies Reality-viable criterion if it elects $R$ whenever $S_R = \emptyset$, and else elects a member of $S_R$.
\end{definition}

In the standard model of Social Choice, the \emph{Condorcet criterion} states that an alternative preferred to all others by a voter majority shall be elected as a winner, called the Condorcet winner. A winner might not exist due to \emph{Condorcet cycles}.
The criterion treats all social states equally, and hence any method for breaking cycles among them may seem arbitrary.
However,  once Reality is recognized as a distinguished social state, it may be used to eschew or break cycles and in either case identify a winner. Several approaches of increasing complexity can take Reality to the task of addressing Condorcet cycles: 
\begin{enumerate}
    \item Elect Reality, that is, retain the status quo: If voters cannot clearly decide what alternative they prefer over Reality, as manifested by a Condorcet cycle among the alternatives preferred over Reality, then retaining Reality is a safe, albeit conservative, bet.
    
    \item Being more permissive, elect an arbitrary winner among all alternatives preferred over Reality.
    
    \item Being less arbitrary, elect the alternative most preferred over Reality.
    
    \item Employ a distance over social states, if available (e.g., if voter preferences satisfy distance constancy and the underlying distance is known), to elect the alternative closest to Reality among those preferred over Reality.
    
    \item Employ one's favorite tournament solution~\cite{moulin1986choosing} to elect a member of the solution to the tournament among the alternatives preferred over Reality.
\end{enumerate}

All these approaches are rephrased below as Reality-aware Condorcet criteria.  They all generalize the classical, Reality-oblivious, Condorcet criterion.
The first one is the following, most conservative criterion, which translates directly into a Reality-aware voting rule.


\begin{definition}[Conservative Reality-aware Condorcet Criterion]\label{definition:Conservative Reality-aware Condorcet Criterion}
Let $S$ be a set of alternatives with $R \in S$ being the Reality. A Reality-aware voting rule satisfies the \emph{Conservative Reality-aware Condorcet criterion} if it elects the Condorcet winner of $S_R$ whenever it exists, and elects $R$ otherwise.
\end{definition}

Next is a less conservative option, which, in the absence of a Condorcet winner among the Reality-viable alternatives, suggests electing some alternative preferred over the status quo.

\begin{definition}[Permissive Reality-aware Condorcet Criterion]\label{definition:Permissive Reality-aware Condorcet Criterion}
Let $S$ be a set of alternatives with $R \in S$ the Reality.  
A Reality-aware voting rule satisfies the \emph{Permissive Reality-aware Condorcet criterion} if it elects the Condorcet winner of $S_R$ if it exists; elects some member of $S_R$ whenever $S_R \neq \emptyset$ and there is no Condorcet winner among $S_R$; and elects $R$ if $S_R = \emptyset$.
\end{definition}

Observe that corresponding \emph{Conservative/Permissive Reality-aware Condorcet voting rules} can be defined directly to adhere to these Conservative/Permissive Reality-aware Condorcet criteria. These Reality-aware Condorcet-consistent voting rules are sufficiently simple to be stated in one sentence and to be computed manually.  For concreteness, the first one is: \emph{Among the alternatives preferred over the status quo, if one is preferred over all others then elect it. Else retain the status quo}.
Indeed, such simplicity means that these rules are easily communicable and therefore might be considered more trustworthy by the voters.

The next criterion suggests to resolve cycles based on the net preference of its members over Reality. 

\begin{definition}[Preference-over-Reality Condorcet Criterion]\label{definition:condorcet preference}
Let $S$ be a set of alternatives with $R \in S$ the Reality.  A Reality-aware voting rule satisfies the \emph{Preference-over-Reality Condorcet criterion} if it elects $R$ whenever $R_S = \emptyset$; elects the Condorcet winner of $R_S$ whenever it exists; and elects a member of $W = \{\argmax_{s \in S_R} N_R(s)\}$ otherwise, where $N_R(s)$ denotes the number of voters preferring $s$ to $R$ minus the number of voters preferring $R$ to $s$.
\end{definition}

The above criterion is not only logically sound but has a psychological appeal, as we discuss next. All alternatives, except Reality, relate to hypothetical social states.  
A comparison of a hypothetical state to Reality is easier mentally and is more trustworthy than a comparison among two hypothetical states, as the latter  requires more hypothesizing and imagining. In the absence of a Reality-viable Condorcet winner, our criterion elects a winner by employing only comparisons with Reality, hence we argue that it is psychologically more sound than criteria that rely on comparisons among hypothetical states.
Concretely, observe that a cycle of hypothetical social states is even less trustworthy. To illustrate this point, assume that a voter's judgment of two hypothetical social states is faulty with probability $p$.  Then, the probability that there are no voter judgment faults in a cycle containing, say, $z$ social states, is $(1 - p)^z$, which diminishes exponentially in $z$.

\medskip

Last, we mention that the above Reality-aware Condorcet criteria can be refined to consider any tournament solution
(e.g., top cycle, uncovered set, etc.). For example, it is possible to consider a Reality-aware Condorcet criterion that requires to select an element of the top cycle of $S_R$ whenever $S_R \neq \emptyset$, and $R$ otherwise. For more background on tournament solutions, we refer the reader to Moulin~\cite{moulin1986choosing}.

%

We do not feel compelled to explore complications of the criteria arising from selecting a winner from a tournament solution at this time, especially since this would entail further comparisons among hypothetical states, which are both laborious to the voters (compared to the Reality-aware Amendment Agenda, see below), less trustworthy, and with what may be a marginal benefit.
In particular, it is possible that all members of a top-cycle are equally preferred over Reality. In such an extreme case, we feel that it is quite justified to resolve the tie among them arbitrarily, as indeed the voters will be equally happy with a transition from Reality to any of them.
 
In case the distance metric is generally known  (e.g., in a specific election satisfying Utility constancy; see Section~\ref{section:domain restrictions}),
the distance from Reality can also be used to break Condorcet cycles, as follows.
  
\begin{definition}[Distance-from-Reality Condorcet Criterion]\label{definition:condorcet two}
Let $S$ be a set of alternatives with $R \in S$ the Reality.  A Reality-aware voting rule satisfies the \emph{Distance-from-Reality Condorcet criterion} if it elects $R$ whenever $R_S = \emptyset$; elects the Condorcet winner of $R_S$ whenever it exists; and elects a member of $W = \{\argmin_{s \in C} d(s, R)\}$ otherwise.
\end{definition}

Observe that the Distance-from-Reality Condorcet criterion defined above does not rely on subjective voter judgments to break cycles, but only on the (presumably) objective distance measure, if available.

\medskip

\begin{remark}
As noted by Arrow~\cite[Page 95]{arrowbooksecond}, the special default role Reality may play in voting on issues does have a natural counterpart when voting on candidates.  To rectify for the lack of a default alternative, Dodgson~\cite{dodgson1873discussion} suggested to add a fictitious protest, or \emph{scarecrow} candidate, to allow voters to protest against an inadequate list of candidates.   Clearly, the winner of an election must at least win over a brainless scarecrow; if not, the elections are annulled. Observe how this view corresponds to our Condorcet criteria defined above, by referring  to a set of candidates containing a default scarecrow candidate, instead of a set of social states that include Reality. 
\end{remark}

\begin{remark}
We note on the relation of our model, and specifically our Condorcet criteria, to Arrow's theorem. Indeed, Gibbard et al.~\cite{gibbard1987arrow} have shown that Arrow's theorem does not hold in a model very much similar to ours:
  In their model there is a distinguished social state which is present in every election.
In particular, they considered the following voting rule, originally proposed by Richelson~\cite{richelson1984social}:
  It considers only the social states that are preferred over the distinguished social state (Reality-viable alternatives in our terminology) and  breaks ties by serial dictatorship.
Indeed, the distinguished social state is used to break Condorcet cycles in a principled way, in particular resulting in a voting rule that is not completely dictatorial, but nevertheless has visible dictatorial properties (as it employs serial dictatorship).

This result suggests considering the relation of voting rules that satisfy our Condorcet criteria to Arrow's theorem. First, observe that the rule considered by Richelson~\cite{richelson1984social} and Gibbard et al.~\cite{gibbard1987arrow} satisfies our Reality-viable criterion (Definition~\ref{definition:reality viable}).

As for our Conservative Reality-aware Condorcet criterion, observe that any voting rule that satisfies this criterion violates Unanimity (i.e., if all voters prefer some social state to another, then the social state shall be preferred to the other in the aggregated vote). To see this, consider a profile in which all voters rank the Reality last --- if there is no Condorcet winner Reality would still be elected.

As for our other Reality-aware Condorcet criterion, we mention that Gibbard et al.~\cite{gibbard1987arrow} have shown that a modification of Arrow's theorem, specifically where Non-Dictatorship is replaced with Anonymity (namely, the voting rule shall produce the same results, independently of reordering the voters), does hold in their model. Thus, while our other Reality-aware Condorcet criterion might subvert Arrow's theorem using similar dictatorial properties as the rule considered by Gibbard et al.~\cite{gibbard1987arrow}, in essence they are bound to either violate Anonymity, Unanimity, or Independent of Irrelevant Alternatives.
\end{remark}

\subsection{Reality-aware Amendment Agenda}\label{section:agenda}

The Amendment Agenda was proposed by Ramon Llull in his 1299
\emph{Des Artes Electionis}~\cite{mclean1990borda}:  Arrange all alternatives in some order, vote the first against the second, the winner of the two against the third, and so on; finally, elect the final winner.
The Amendment Agenda is known to elect a Condorcet winner if there is one, however it does not detect cycles.  

Here we use this Amendment Agenda as the core of our Reality-aware Condorcet-consistent voting processes, which corresponds naturally to our Reality-aware Condorcet criteria.  All processes consist of two stages: (i) Identify the Reality-viable alternatives and perform the Amendment Agenda on them; (ii) Vote the final winner against all tentative winners (except the previous one, against which it has already won).  If it wins them all, then elect it. 

The specific variants of our adaptation of the Amendment Agenda to our setting differ in how they proceed in cases in which there is no Condorcet winner among the Reality-viable alternatives.
E.g., the conservative process, which is formally described below, simply elects Reality in this case.  

\begin{algorithm}[Reality-aware Conservative Amendment Agenda]\label{algorithm:RAAA}
Let $S$ be the set of alternatives with $R \in S$ and let $S_R$ be the set of Reality-viable alternatives.  If $S_R = \emptyset$ then elect $R$.  Else perform an Amendment Agenda on $S_R$ starting with $R$, followed by extra votes of the final winner against all other members of $S_R$ not previously voted against, if any. If the final winner wins all these votes then elect it.
Else elect $R$.
\end{algorithm}

\begin{remark}\label{remark:agendas}
Amendment Agendas corresponding to the other Reality-aware Condorcet criteria described in the previous section can be obtained by revising the final ``Else".  For example, revise the final ``Else'' to ``Else arbitrarily elect a member of $S_R$.'' to obtain the Permissive Reality-aware Amendment Agenda; or revise it to ``Else elect from $S_R$ the alternative that wins the most over $R$, breaking ties arbitrarily.'' to obtain the Preference-over-Reality Amendment Agenda. Further complicating the Agenda can produce processes that correspond to other criteria, e.g. those 
based on tournament solutions.
\end{remark}

\begin{theorem}
  Algorithm~\ref{algorithm:RAAA} satisfies the Conservative Reality-aware Condorcet criterion. Furthermore, the modifications described in Remark~\ref{remark:agendas} correspond to algorithms which satisfy the further Reality-aware Condorcet criteria discussed in Section~\ref{section:condorcet}. 
\end{theorem}

\begin{proof}
We prove the claim only for Algorithm~\ref{algorithm:RAAA}, as the proof is similar for the other agendas.
There are two cases to consider, one in which there is a Condorcet winner $c$ among $S_R$ and two in which there is not.  In the first case, assume by way of contradiction that that $c' \ne c$ is elected. But this would require $d$ to beat $c$ either during one of the  Amendment Agenda votes on $S_R$ performed as the first step of Algorithm~\ref{algorithm:RAAA}, if $c$ was tentatively elected before $d$, or in the second step of voting it against all other alternatives of $S_R$, final vote if it was not, contradicting $c$ being a Condorcet winner.  Thus, in this case Algorithm~\ref{algorithm:RAAA} will select $c$ as the winner,
as required by the Conservative Reality-aware Condorcet criterion.

In the second case, let $d$ be the alternative selected in the first step of Algorithm~\ref{algorithm:RAAA}, namely the winner of the Amendment Agenda on $S_R$.
As $d$ is in particular not a Condorcet winner of $S_R$, it will fail against one of the alternatives it is voted against in step two of Algorithm~\ref{algorithm:RAAA}, thus $R$ will be selected, as required by the Conservative Reality-aware Condorcet criterion.
\end{proof}

\begin{remark}
We wish to clarify the use of the term  Reality in the description of the  Reality-aware Amendment Agenda described in Algorithm~\ref{algorithm:RAAA}:
Here, Reality is best understood as a variable, with an initial value being the present reality, namely the status quo.  As the algorithm progress, the value of the variable changes.  The final value is the elected alternative, which is the new Reality.
\end{remark}

\section{Applications}\label{section:applications}

Since the introduction of key ideas presented here in the preliminary version of this work~\cite{preliminary}, several applications of Reality-aware Social Choice have been demonstrated. Here we briefly discuss them, and other, envisioned ones.



\subsection{Sybil-Resilient Voting}

In a majoritarian democracy, a single vote may decide the fate of elections or tilt a decision. Such a deciding vote could, in principle, belong to a fake voter, aka sybil. The risk of sybils infiltrating the electorate is even higher in online democratic communities, making sybil attacks literally an existential threat to e-democracies. 
Recently, Reality-Aware Social
Choice was leveraged to face sybil attacks~\cite{srsc}. Specifically, Shahaf et al.~\cite{srsc} use the status quo as the anchor of \emph{sybil resilience}, which is characterized by \emph{sybil safety} - the inability of sybils to change the status quo against the will of the genuine voters, and \emph{sybil liveness} – the ability of the genuine voters to change the status quo against the will of the sybils.
They investigate the supermajority needed to ensure sybil safety, as a function of the sybil penetration rate, and show that liveness can be assured up to one-third sybil penetration.

Notably, the use of Reality-aware Social Choice as the foundation of their sybil-resilient voting rules allows Shahaf et al. to bypass known lower bounds~\cite{conitzer2010using}.

\subsection{Holistic Voting Process}

When several alternatives to the status quo are offered, and the majority prefers the status quo over any of them, then, according to the Reality-aware Condorcet criteria discussed in Section~\ref{section:condorcet}, the status quo shall remain.
Indeed, it is possible that an alternative with majority preference over the status quo does not exist. But, it could also be the case that such an alternative exists but was not offered to vote upon, due to ignorance, incompetence or the special interests of those who set up the vote.

Such a situation is clearly distressing, as the will of the people to change the status quo cannot be accommodated. A recent paper~\cite{metricaggregation}, inspired by Reality-Aware Social Choice, describes an egalitarian process for identifying alternatives that may have majority support over the status quo in a broad range of social choice settings, including ordinal elections, committee elections, participatory budgeting and participatory legislation.

Such omission of alternatives is less likely to happen in the standard setting of social choice, where the objects voted upon are usually quite simple, e.g., candidate names. In other settings of social choice, for example, in participatory budgeting~\cite{goel2015knapsack,aussieone,abpb}, or when deciding upon a legislation, on in participatory legislation~\cite{alsina2018birth}, the objects upon which the voters vote on are more complex, and hence offering a closer list of alternatives to vote upon is more prone to the omission of majority-supported alternatives.
Such complex social choice settings might benefit from a more holistic process than merely voting on the available alternatives. Specifically, Reality-Aware Social Choice can be utilized with the following idea. For concreteness, we speak of a society wishing to amend a legislation.

Here, in a way which resembles the idea of a democratic action plan discussed in Section~\ref{section:action plans}, the society might start from the status quo, i.e., the current legislation on the topic, which we term the common draft. Then, a deliberation of possible edits and revisions to the common draft takes place, followed by a vote, in which participants in the society vote on the suggested revisions. The result of the vote would be a revised aggregated common draft, which incorporates, say, those edits to the previous common draft for which there is majority support.

While the newly aggregated common draft might be the end result, utilizing Reality-Aware Social Choice we may let it be the new common draft (Reality), and repeat the round of deliberation and voting. This can be continued until some time limit is reached, or, preferably, until there is consensus, or at least majority preference over the status quo, among the participants with regard to the resulting common draft.
The work done by Shahaf et al.~\cite{metricaggregation} is a first step towards a realization of such a holistic voting process.

\section{Discussion}\label{section:discussion}

Augmenting the standard model of social choice by introducing Reality as a distinguished alternative necessitates building a new foundation for this theory.
We have proposed a simple model of Reality-Aware Social Choice and considered four domain restrictions on voter preferences in this model. We have argued that Reality can then be used as a principled way for breaking Condorcet cycle, described several corresponding Reality-aware Condorcet-criteria and proposed an efficient Reality-aware Condorcet-consistent agenda.
We find it particularly satisfying that more than seven hundred years after Ramon Llull has proposed a principled show-of-hands agenda~\cite{colomer2013ramon},
it is still possible to offer a novel and improved computer-less agenda. 

Further study of the foundations of Reality-Aware Social Choice shall be conducted, especially in the face of its many applications, which include decision making mechanisms for elections~\cite{metricaggregation}, budgeting~\cite{budgetingone}, and legislation~\cite{textvoting}, and sybil-resilient voting rules~\cite{srsc}.

\bibliography{bib}
\bibliographystyle{plain}

\end{document}